\documentclass[11pt]{article}

\usepackage[letterpaper,margin=1in]{geometry}
\usepackage{fullpage}
\usepackage{amsmath}
\usepackage{amsthm}
\usepackage{multirow}
\usepackage{footmisc}
\usepackage{amssymb}
\usepackage{graphicx}
\usepackage{color}
\usepackage{hyperref}
\usepackage{enumitem}

\usepackage{algorithm}
\usepackage{algpseudocode}
\usepackage[linesnumbered,ruled,vlined,algo2e]{algorithm2e}

\makeatletter
\def\namedlabel#1#2{\begingroup
    #2%
    \def\@currentlabel{#2}%
    \phantomsection\label{#1}\endgroup
}

\newcommand{\framework}{\textbf{\texttt{refix-leader}}}
\newcommand{\ul}{\textsc{unique-leader}}

\newtheorem{theorem}{Theorem}
\newtheorem{lemma}{Lemma}

\newtheorem{definition}{Definition}

\newtheorem{remark}{Remark}
\newtheorem{claim}{Claim}

\begin{document}

\title{Short Graph Sketches Suffice for \\ Error-resilient Leader Verification in CONGEST}

\author{
Pawe\l\ Garncarek$^{1}$\thanks{\texttt{pgarn@cs.uni.wroc.pl}},
Tomasz Jurdzi\'nski$^{1}$\thanks{\texttt{tju@cs.uni.wroc.pl}},
Dariusz Kowalski$^{2}$\thanks{\texttt{darek.liv@gmail.com}},
Subhajit Pramanick$^{1}$\thanks{\texttt{suvo.iitg17@gmail.com}}
}

\date{
$^{1}$Institute of Computer Science, University of Wroc\l aw, Poland\\
$^{2}$Department of Computer Science, Augusta University, USA
}

\maketitle

\begin{abstract}
Locally Checkable Proofs (LCPs) enable the verification of global graph
properties from locally checkable certificates assigned by a prover. This
framework has recently been extended to Locally Checkable Proofs-with-Errors
(LCPE), in which an adversary may corrupt some of the certificates. Existing
LCPE algorithms are designed for the LOCAL model, whose unbounded
communication makes them unsuitable for direct implementation in the
(logarithmically) bandwidth-restricted CONGEST model.

In this paper, we initiate the study of efficient CONGEST
implementations of LCPE through the \textsc{unique-leader} verification problem
on trees. The main challenge is that tolerating $\varepsilon$ certificate errors
requires each node to reason about its $(2\varepsilon+1)$-hop neighborhood, whose
exact topology cannot, in the worst case, be communicated efficiently in
\textsc{Congest} -- reconstructing it naively costs up to
$O(\Delta^{2\varepsilon+1}\log n)$ bits. To overcome this bottleneck, we
introduce \emph{local graph sketches}, together with the notions of
\emph{imagined trees} and \emph{imagined certifications}, which compactly
encode---in only $O(\varepsilon^2\log n)$ bits at each node---precisely the
information a node needs to decide. Using these techniques, we design an
algorithm that tolerates up to $\varepsilon$ adversarial certificate errors and
computes the required sketches in $O(\varepsilon^2)$ communication rounds in the
CONGEST model.

We complement this with a matching impossibility result: even in the strictly
more powerful LOCAL model, and even with certificates of unbounded size,
no verification scheme with view distance at most $\varepsilon$ can tolerate
$\varepsilon$ adversarial errors. Since LOCAL is strictly stronger than
CONGEST, this lower bound carries over immediately, showing that a view
distance beyond $\varepsilon$ -- and hence the wider neighborhood our algorithm
summarizes -- is unavoidable.

\vspace{0.2cm}
\end{abstract}
\newpage


\section{Introduction}
\label{sec:introduction}

\subsection{Background and Motivation}
\label{subsec:background}
In general, nodes in a distributed system do not have a global view of the network. 
Instead, each node can access only a bounded neighborhood around itself (sometimes only its direct neighborhood), whereas many properties of interest are inherently global. 
Examples include determining whether the network is a tree, whether a graph is bipartite, whether a computed structure is a valid spanning tree, or whether the system remains consistent after transient faults.
Such verification tasks arise in many settings, including self-stabilizing and fault-tolerant systems. For many global properties, however, purely local verification is impossible without additional information, while verifying them from scratch requires excessive communication. This motivates equipping each node with a small amount of auxiliary information that makes local verification possible. 

The notion of locally verifying graph properties with the help of auxiliary information assigned to the nodes is captured by the framework of \emph{proof labeling schemes} (PLS), introduced by Korman, Kutten, and Peleg \cite{10.1145/1073814.1073817}. In this framework, a \emph{prover} first assigns a certificate (a binary string) to each node. Subsequently, a distributed verification algorithm, called the \emph{verifier}, runs at every node and uses its own certificate, together with the certificates of the neighboring nodes, to output a binary decision: either \emph{accept} or \emph{reject}.
The aim of the verifier is to distinguish between the yes and no instances of the verification problem.
Later G{\"o}{\"o}s and Suomela \cite{goos2016locally} generalized the idea of proof labeling schemes under the name of \emph{locally checkable proofs} (LCP), where the main difference is that a node (rather the verifier) can inspect a constant radius (instead of just 1) neighborhood and that the node can access the certificates and identifier of the nodes within that neighborhood.
This constant radius is called the \emph{view distance} of the node.
A survey by Feuilloley \cite{feuilloley2021introduction} consolidates these notions.
This framework on local verification has been extended in various directions, including PLSs with restricted provers \cite{DBLP:conf/wdag/EmekGK22}, randomized \cite{10.1145/2767386.2767421} and quantum PLSs \cite{fraigniaud_et_al:LIPIcs.ITCS.2021.28}, PLSs with global proofs in addition to local ones \cite{feuilloley_et_al:LIPIcs.DISC.2018.25}, and many more.

The model underlying this long line of work is the well-known LOCAL model \cite{peleg2000distributed}, in the sense that a verifier with a view distance $k$ can be implemented in $k$ synchronous communication rounds (see \cite{goos2016locally} by G{\"o}{\"o}s and Suomela) of the LOCAL model. 
In this context, the metric of interest is the certificate size (i.e., the maximum number of required bits given to any node) considered across the literature for various graph properties on different families of graphs \cite{bousquet_et_al:LIPIcs.DISC.2025.18,10.1145/3519270.3538416,fraigniaud_et_al:LIPIcs.DISC.2023.20,DBLP:journals/algorithmica/FraigniaudMRT24,garncarek2026distributed}.
The certificate size depends on the property being verified.
Constant-size is known for several problems \cite{goos2016locally}, whereas an $O(\log n)$ size suffices for a wide range of verification problems, including unique leader, cycle-freeness, spanning tree, non-bipartiteness, minor-freeness and many more \cite{10.1007/978-3-319-12340-0_2,bousquet2024local,ESPERET202268,10.1145/3382734.3404505,FEUILLOLEY20239,feuilloley_et_al:LIPIcs.DISC.2018.25,feuilloley2025proving,10.1145/2499228,goos2016locally,ostrovsky2017space}.
On the other hand, some verification problems require larger proofs, such as $O(\log n\log W)$ bits (edge weights in $[1, W]$) for MST verification by Korman and Kutten \cite{DBLP:journals/dc/KormanK07} and $O(n\log n)$ bits for diameter $k$ verification by Censor-Hillel, Paz and Perry \cite{CENSORHILLEL2020112}.

\paragraph{Erroneous Certificates.} In this context, always expecting trustworthy certificates may be limiting in practice. Since certificates must first be generated, they may be produced using AI/ML-based techniques or randomized algorithms (e.g., Monte Carlo algorithms), both of which may introduce errors.
In our context, an error could refer to a certificate that differs from the one assigned by a trustworthy prover.
More generally, the verifier has no control over how the certificates are generated or transmitted, and therefore cannot always rely on their correctness. Consequently, it is natural to seek verification schemes that continue to operate correctly even when (possibly a bounded number of) certificates are erroneous. This motivates the study of verification schemes that remain robust even when the certificates are partially inaccurate.
Similar motivations have recently appeared in several works on learning-augmented algorithms, where the auxiliary information provided to the nodes may be imperfect or unreliable \cite{10.1145/3732772.3733518,10.1145/3732772.3733530,lykouris2021competitive}. 
This direction has recently attracted considerable attention across several research communities.
In fact, Boyar, Ellen, and Larsen in \cite{10.1145/3732772.3733530} initiated the study of distributed graph algorithms in this setting, although their model (which they refer to as the prediction model) differs somewhat from ours in this paper.

\paragraph{LOCAL vs CONGEST.} 
To the best of our knowledge, none of these works considers the CONGEST model -- one of the classical models in distributed computing, which assumes per-edge bandwidth. The LOCAL model, on the other hand, places no limit on message size, and thus a node can collect its full $k$-hop neighborhood in $k$ rounds, whereas the CONGEST model caps every edge at $O(\log n)$ bits per round.
This gap is largely immaterial for certification so far, since the verifier for essentially every graph property studied in the literature inspects only a node's immediate ($1$-hop) neighborhood.
The situation changes once we allow an adversary to tamper with the certificates that the (trustworthy) prover assigns.
With such corruption, a $1$-hop view can be too unreliable. 
This motivates a more general formulation, in which the adversary may modify the certificates of at most $\varepsilon$ nodes within the $k$-hop neighborhood of each node.
Now, a node can no longer rely on its immediate neighborhood alone and must examine a wider neighborhood.
Collecting a $k$-hop neighborhood costs merely $k$ rounds in LOCAL, but the $k$-hop neighborhood contains $O(\Delta^k \log n)$ bits of information, which in CONGEST model, requires round complexity up to $O(\Delta^k)$ in the worst case ($\Delta$ is the maximum degree).
The central challenge of this work is to verify the desired property correctly 
using substantially smaller $O(\varepsilon^2 \log n)$ information
than this naive $O(\Delta^k \log n)$ bound.


\subsection{Our Model and Definitions}
\label{subsec:model}

\paragraph{Graphs.} A graph $G(V,E)$ in our consideration is a tree with $|V| = n$ nodes.
For a node $v$, we use the notation $N_{k}(v)$, called the \emph{$k$-hop neighborhood of $v$}, to denote the set of all nodes within the distance $k$ hops from $v$ (the shortest path from $v$ to any node in $N_k(v)$ has  $\leq k$ edges), including $v$ itself.
We assume that the nodes of a tree possess unique identifiers (or IDs for short) from the set $\{1,2, \dots, poly(n)\}$ on $O(\log n)$ bits.
However, our techniques do not rely on the availability of node IDs. 
It is sufficient to assume that each node $v$ has a distinct local port numbering from the set $\{1,2, \dots, deg(v)\}$ to its incident edges. 
Note that the port numbers assigned to the two endpoints of an edge are independent and need not coincide.

\paragraph{Local Verification with Certificates.} 
A convenient way to describe the verification process is through a \emph{prover-verifier} paradigm. 
The prover gives a \emph{certificate assignment}\footnote{In some of the previous works, this certificate assignment function is also called a \emph{proof} on $G$.} $C$ on a graph $G$, which is a function $C: V \rightarrow \{0,1\}^*$ that associates with every node a \emph{binary string}.
A binary string assigned to a node $v$ is called the \emph{certificate of} $v$. 
The prover's objective is to have all nodes accept, irrespective of whether the instance is a yes-instance or a no-instance.

A \emph{verifier} is a distributed algorithm $\mathcal{A}$ running at each node that takes the triple $(G, L, v)$ as input and outputs either $0$ or $1$.
For a natural number $k$ and a node $v \in G$, let $G[N_k(v)]$ be the subgraph induced by the nodes in $N_k(v)$ and $C[N_k(v)]: N_k(v) \rightarrow \{0,1\}^*$ be the restriction of $C$ to $N_k(v)$.
A verifier $\mathcal{A}$ is called a \emph{local verifier} if there exists a $k$ such that $\mathcal{A}(G, L, v) = \mathcal{A}(G[N_k(v)], L[N_k(v)], v)$ for all $G, L$ and $v$, which intuitively means that the output of a node $v$ depends only the information (inputs or certificates) of the nodes within its $k$-hop neighborhood.
We call this $k$ the \emph{view distance} of $v$ (which, in the distributed setting, is equivalently the \emph{view distance} of $\mathcal{A}$).

In the context of verifying a graph property (recall that a graph property is a set of graphs that is closed under isomorphism), we say that a graph property $\mathcal{P}$ admits a \emph{locally checkable proof labeling scheme} if there is a prover-verifier pair $(f,\mathcal{A})$, where the prover $f$ gives a certificate assignment function $C$ for each $G \in \mathcal{P}$ and $\mathcal{A}$ is a local verifier (with a view distance $k$) such that the following two properties hold.
\textbf{\textit{Completeness}:} If $G \in \mathcal{P}$, there should be a certificate assignment function $C$ on $G$ provided by the prover $f$ (i.e., $f(G) = C$) such that $\mathcal{A}(G[N_k(v)], C[N_k(v), v]) = 1$ for all $v$.
\textbf{\textit{Soundness}:} If $G \notin \mathcal{P}$, then for every certificate assignment $C$ on $G$, there must be at least one node $v$ such that $\mathcal{A}(G[N_k(v)], C[N_k(v)], v) = 0$.

For our convenience, we say that the node $v$ \emph{accepts} (resp. \emph{rejects}) when the verifier $\mathcal{A}$ outputs $1$ (resp. $0$).
We emphasize the fact that when $G$ does not satisfy the property $\mathcal{P}$, there must exist at least one rejecting node for any certificate assignment function.

\paragraph{Oracle vs Adversary (The Error Model).} 
To introduce the concept of certificates with errors in a convenient way, we distinguish two entities that assign certificates to nodes: the \emph{oracle} and the \emph{adversary}.
So far in the above discussion of the model, the prover $f$ is a trustworthy oracle.
We denote the certificate assignment function provided by an oracular prover by $C_O$ and the corresponding verifier by $\mathcal{A}_0$.
We refer to the pair $(C_O, \mathcal{A}_0)$ as the \emph{oracular scheme}.
 
On the other hand, we consider a setting in which an adversary can modify the certificates assigned to the graph nodes by an oracular prover to potentially cause all nodes to accept even when $G\notin \mathcal{P}$ or some nodes to reject when $G \in \mathcal{P}$.
We represent the certificate assignment function after the adversarial modification by $C_{adv}$.

In our paper, we assume that $C_{adv}$ differs from $C_O$ in at most $\varepsilon$ nodes within $N_{2\varepsilon+1}(v)$ for every node $v$, i.e., the adversary is allowed to introduce at most $\varepsilon$ modifications within the $(2\varepsilon+1)$-hop neighborhood of each node $v$.
We call each such modification an \emph{error}.
Under this adversarial model, the definition of completeness must be modified to account for certificates with errors, whereas the definition of soundness remains unchanged. 
 \begin{itemize}
     \item \textbf{\textit{Error-resilient Completeness}:} If $G \in \mathcal{P}$, there is a certificate assignment $C_O$ on $G$ such that for every $C_{adv}$ differing from $C_O$  in at most $\varepsilon$ nodes within $N_{2\varepsilon+1}(v)$ for every node $v$, we have $\mathcal{A}(G[N_{k}(v)], C_{adv}[N_k(v)], v) = 1$ for all $v$.
 \end{itemize}

In this setting, we refer to the pair $(C_{adv}, \mathcal{A})$ as the \emph{adversarial scheme}.
In our description, we use the term \emph{\textbf{the oracular setting}} to mean the error-free setting in which certificates are not modified by the adversary, and the term \emph{\textbf{the adversarial setting}} to mean the setting in which certificates might be influenced by the adversary.

\subsection{The Problem Definition: \textmd{\ul{}}}

In this paper, we consider the popular verification problem of determining whether a graph has a unique leader \cite{feuilloley2025proving,goos2016locally}.
For simplicity, we assume that each node is initially assigned a binary input, rather than a unique identifier. These inputs are part of the problem instance and should not be confused with certificates, which are assigned later by the prover.
The leader node in the graph gets $\mathcal{L}$, and the rest of the nodes get $\mathcal{N}$.
In this paper, we study the property $\mathcal{P}=\ul{}$. The corresponding yes and no instances are defined as follows. An instance is a yes-instance for \ul{} if exactly one node in the graph is assigned input bit $\mathcal{L}$. An instance is a no-instance whenever either no node or more than one node is assigned input bit $\mathcal{L}$.
We say that a scheme $(C, \mathcal{A})$ solves \ul{} if both completeness and soundness are satisfied.

\paragraph{Warm-up: \textmd{\ul} in Oracular Setting.}
It is straightforward to verify \ul{} using the following known scheme, which we mention briefly.

The certificate assignment function $C_O(v) = 0$ if $v$ is the node with the input bit $1$ (leader).
For any other node $v$ (non-leader), $C_O(v)$ is the breadth-first search (BFS) distance to the leader.  

We use $\mathcal{A}_0$ to denote the verification algorithm executed at each node $v$ that uses the certificates of $v$ and its neighboring nodes.
More precisely, a node $v$ with input bit $0$ checks whether exactly one neighbor, namely its \emph{parent}, has certificate $C_O(v)-1$, while every other neighbor, namely its \emph{children}, has certificate $C_O(v)+1$.
If $v$ is the leader (i.e., the unique node with input bit $1$), it has no parent, and therefore all of its neighbors must have certificate $C_O(v)+1$. A node accepts if the corresponding condition is satisfied; otherwise, it rejects.

The correctness of the above scheme is well known and straightforward.
Hence, we do not include it in our paper. 
As $\mathcal{A}_0$ requires only the certificates of neighboring nodes, the distinction between the LOCAL and CONGEST models is immaterial in this setting.
Also, observe that the above scheme uses certificates of size $O(\log n)$.
In fact, this certificate size is optimal for an arbitrary graph with a matching lower bound, established by G{\"o}{\"o}s and Suomela \cite{goos2016locally}.
We refer to $\mathcal{A}_0$ as the \emph{base verification algorithm} henceforth (when $\varepsilon = 0$).

\subsection{Our Results and Roadmap}
\label{subsec:our-results}

We initiate the study of efficient CONGEST implementations of error-resilient local verification on the \ul{} problem in this paper.
We begin, in Section \ref{sec:impossibility}, with an impossibility result (see Theorem~\ref{thm:impossibility}) which shows that tolerating $\varepsilon$ errors, even if $\varepsilon$ errors are spread across the entire graph (rather than within the neighborhood of every node) forces every node to examine a neighborhood strictly larger than $N_{\varepsilon}$ and this holds already in the LOCAL model, which is strictly more powerful than CONGEST. 
The lower bound illustrates the need to learn large neighborhoods to solve problems in the context of error-resilience. In the CONGEST model, learning the entire $\varepsilon$-hop neighborhood may need $\Omega(\Delta^\varepsilon)$ rounds.

Our positive results circumvent this problem.
We present an algorithm that requires nodes to learn the certificates of nodes within their $2\varepsilon+1$-hop neighborhood to verify \ul{} and we separate it into two parts: \textit{learning} and \textit{decision}. 
The key observation is that the decision part never requires the exact structure of the neighborhood: a compact summary suffices. 
Hence, in the learning part, we introduce a novel technique, which we refer to as (local) \emph{graph sketches}, which is used to describe all relevant information about its $2\varepsilon+1$-hop neighborhood in a compact form that a node can learn in $O(\varepsilon^2)$ CONGEST round.
The construction is given as \textsc{sketch-construction}() (Algorithm~\ref{alg:sketch}) in Section \ref{sec:sketches}.

The decision part is our algorithm \framework{} (Section~\ref{subsec:framework}): it takes the compressed neighborhood information, searches for a correction of at most $\varepsilon$ certificates under which the error-free base verifier accepts, and decides in a single round of local computation. 
The correctness of \framework{} hinges on one lemma (Lemma ~\ref{prop:diff_view_dist}). As its proof is long and largely self-contained, we state and use it in the analysis but defer the proof to Section~\ref{sec:erroneous-labeling} for space management.
We conclude in Section~\ref{sec:conclusion} with open problems and future directions.

Although learning logically precedes decision, we present them in the opposite order: the decision algorithm first (Section~\ref{subsec:framework}), then the sketch construction (Section~\ref{sec:sketches}).
This lets us first fix exactly what information the decision part consumes --- the interface the sketch must provide --- and only then show how to compute it, which we find considerably clearer than specifying the compression before its purpose is known.
We believe the graph-sketch technique is of independent interest and opens several directions for verification problems. The broad idea of summarizing a graph compactly to compute on it efficiently has been studied before, in settings different from ours~\cite{ahn2012analyzing,ghaffari2018congested,jurdzinski2018mst}; we expect such compression to find further use in communication-efficient, error-resilient solutions based on certificates, and in the advice model.

For the reader's convenience, we include Table~\ref{table:variables} at the end of the paper, summarizing the main variables and notations used throughout.

\section{Impossibility Result}
\label{sec:impossibility}

We establish the following impossibility result in the LOCAL model. Since LOCAL is strictly more powerful than CONGEST, the lower bound immediately carries over to the CONGEST model.
This impossibility holds even when the adversary may modify the certificates of at most $\varepsilon$ nodes in the entire graph (rather than within the neighborhood of every node, which would allow for more than $\varepsilon$ errors in the entire graph distributed among multiple neighborhoods).

\begin{theorem}\label{thm:impossibility}
In adversarial setting, there is no scheme $(C_{adv},\mathcal{A})$ solving
\textsc{unique-leader} on arbitrary trees with $\varepsilon$ errors in which the verifier
$\mathcal{A}$ has view distance at most $\varepsilon$
even if certificates are allowed to have an unbounded size.
\end{theorem}

\begin{proof}
Assume for the sake of contradiction that such a scheme
$(C_{adv},\mathcal{A})$ exists, where $\mathcal{A}$ has view distance $\varepsilon$.
By error-resilient completeness, for every yes-instance, there is an
oracular certificate assignment $C_O$ under which every node accepts,
and continues to accept under every assignment obtained by modifying at
most $\varepsilon$ certificates within the $\varepsilon$-hop neighborhood of each node.

\medskip
\noindent\emph{The yes-instance.}
Let $G_Y=(v_0,v_1,\dots,v_{4\varepsilon+2})$ be a path on $4\varepsilon+3$ nodes, where
$v_0$ is the unique leader (input bit $\mathcal{L}$). 
By completeness, fix an oracular certificate assignment $C_O$ under
which every node of $G_Y$ accepts and keeps accepting under any
admissible adversarial modification.
For $0\le j\le 4\varepsilon+2$, let $c_j=C_O(v_j)$.
Their actual values are irrelevant to this proof, and we use them only to construct the certificate assignment for the second tree. 

We highlight a property we need.
We say that a node $v$ in a tree $G$ has a \emph{yes-view} if there exists a tree $G'$, a node $v' \in G'$ and an oracular certification $C_O'$ on $G'$ such that views of $v$ and $v'$ are identical in tree structure and input bits, and the certification within the view of $v$ differs in at most $\varepsilon$ nodes from $C_O'$ on nodes within the view of $v'$. 
By error-resilient completeness, node $v$ must accept if it has a \emph{yes-view}.


\medskip
\noindent\emph{The no-instance.}
Let $G_N=(u_0,u_1,\dots,u_{4\varepsilon+2})$ be a path on the same number of
nodes, but with \emph{two} leaders: $u_0$ and $u_{4\varepsilon+2}$ have input bits
$\mathcal{L}$, and every other node has input bit $\mathcal{N}$. 
We now define a certificate assignment $C$ on $G_N$ by copying $C_O$ inward from each end, reflected about the midpoint (central node $u_{2\varepsilon+1}$) of the path:
$C(u_j)=c_{\min\{j,\,4\varepsilon+2-j\}}, $ where $ 0\le j\le4\varepsilon+2$

Fix a node $u_m$. 
If $m\le 2\varepsilon+1$, then we compare $u_m \in G_N$ and $v_m \in G_Y$; otherwise we compare $u_{4\varepsilon+2-m} \in G_N$ and $v_m \in G_Y$. These cases are analogous, so we present only the case for $m\le 2\varepsilon+1$.
We show that $u_m$ has a yes-view by comparing it to the view of $v_m$ under the oracular certification $C_O$.
The view of $u_m$ consists of the nodes $u_{\max\{0, m-\varepsilon\}},\dots, u_{\min \{4\varepsilon+2, m+\varepsilon\}}$, and likewise the view of $v_m$ consists of $v_{\max\{0, m-\varepsilon\}},\dots, v_{\min \{4\varepsilon+2, m+\varepsilon\}}$.
Since $m\le2\varepsilon+1$, the largest index appearing in either view is $m+\varepsilon \le 3\varepsilon +1 < 4\varepsilon+2$. Hence, neither view reaches the right endpoint of the path.
In particular, the right leader $u_{4\varepsilon+2}$ of $G_N$ never lies in $u_m$'s view. 
The two views thus span the same range of indices and have the same path structure.
Moreover, if $m \geq \varepsilon+1$, neither view contains index $0$, so neither contains a leader. 
If $m\leq \varepsilon$, both views contain index $0$ and no other leader index, and $u_0 \in G_N$ and $v_0 \in G_Y$ are both leaders.  
Either way, the leaders align on both views. 
Hence, the two views agree in structure and in leader placement, and can differ only in certificates. 
By construction $C(u_j) = c_j = C_O(v_j)$ for all $j \leq 2\varepsilon+1$, so a difference can occur only at an index $j > 2\varepsilon+1$. 
But $u_m$'s view reaches only up to index $m+\varepsilon$, and since $m \leq 2\varepsilon+1$, it contains at most $(m+\varepsilon)-(2\varepsilon+2)+1 \;=\; m-\varepsilon-1 \;\le\; \varepsilon$ indices above $2\varepsilon+1$. 
Hence, the two views differ in at most $\varepsilon$ certificates, so $u_m$ has a yes-view and
$\mathcal{A}$ makes $u_m$ accept.

Consequently $\mathcal{A}$ accepts at every node of $G_N$, even though
$G_N$ has two leaders and is a no-instance of \textsc{unique-leader}.
This contradicts soundness and completes the proof.
\end{proof}

\section{\textmd{\ul} in Adversarial Setting (with Errors)}

    \label{subsec:framework}


The impossibility of Section~\ref{sec:impossibility} shows that, in the adversarial setting, every node must inspect a neighborhood of radius at least $\varepsilon+1$ to decide correctly. 
The algorithm we develop in this section requires a view distance of $2\varepsilon+1$.
One of the main challenges here is due to CONGEST -- how can a node learn about its $2\varepsilon+1$-hop neighborhood within $poly(\varepsilon)$ rounds?

We split the task into two logical parts. The first is \emph{learning} the neighborhood: rather than reconstructing the exact $(2\varepsilon+1)$-hop neighborhood, each node learns a compact summary of it, which we call a \emph{graph sketch}. 
This construction succeeds only when the certificates a node sees are within $\varepsilon$ corrections of a certification that the base algorithm $\mathcal{A}_0$ would accept; if no such correction exists, the construction cannot produce a valid sketch, and we say the node encounters a \emph{sketch rejection}. A sketch rejection can arise only when the node ought to reject in the first place --- either the instance has no unique leader, or the certificates are too corrupted to be fixed --- so a node that encounters one simply rejects. It therefore suffices, for the rest of this section, to reason about nodes that obtain a valid sketch.
We defer the construction to Section~\ref{sec:sketches}.
In the present section, we assume each node has already obtained its graph sketch; we state precisely what information the sketch provides and use it to decide locally.
The second part is thus \emph{deciding}: given the sketch, each node runs \framework{}, which we present and prove correct in Sections~\ref{subsec:analysis}--\ref{sec:erroneous-labeling}. Only afterward, in Section~\ref{sec:sketches}, do we show how the data must be transmitted to compute the sketch.

\subsection{General Idea of the Algorithm \framework{}}
\label{subsec:general-idea-algorithm}

We begin with a \emph{graph sketch} $S_G^v$, from which a node $v$ obtains the following information. Although the sketch (of a particular size) may encode additional information, we list only the information required by our algorithm.

\begin{enumerate}[leftmargin=*]
\setlength{\itemsep}{-1pt}
    \item The number of nodes in $N_{2\varepsilon+1}(v)$ in $G$.    
    \item For each $u \in N_1(v)$, the adversarial label $C_{adv}(u)$ together with the port at $v$ leading to node $u$.
    

    \item[\namedlabel{itm:one}{3}.] For every sequence of certificates $(c_1, c_2, \dots, c_k)$ with $k \le 4\varepsilon+3$, every pair of positions $m,j \in \{1, \dots, k\}$, and every pair $u_1, u_2$ of neighbors of $v$, the node $v$ knows whether the sequence is \emph{realizable at positions $m,j$ through $u_1, u_2$}: that is, whether there exists a path $P = (v_1, v_2, \dots, v_k)$ contained entirely within $N_{2\varepsilon+1}(v)$ such that (i) $v$ occupies position $m$ in $P$, i.e.\ $v = v_m$, (ii) a node $x$ with input $\mathcal{L}$ occupies position $j$ in $P$, i.e.\ $x = v_j$, (iii) $u_1 = v_{m-1}$ whenever $m \ge 2$, and $u_2 = v_{m+1}$ whenever $m \le k-1$ (the path enters and leaves $v$ through the specified neighbors) and (iv) the adversarial certificates along $P$ match the sequence, i.e.\ $C_{adv}(v_i) = c_i$ for every $i$.
\end{enumerate}

The procedure by which a node constructs (or obtains) the graph sketch of its $N_{2\varepsilon+1}$ is described in Section~\ref{sec:sketches} -- for now, we assume a node has it and describe how it decides. 
A node might not learn the exact structure of the neighborhood from the sketch, so it instead considers every neighborhood that is consistent with what the sketch records. 
We formalize this next with the following definitions: the neighborhood a node imagines, the certificates it carries, and the corrections the node finds.

\begin{definition}[\textbf{Imagined Tree}]
    \label{def:imagined-graph}
    Let $S_G^v$ denote the graph sketch that node $v$ obtains in $G$ -- which captures only part of $N_{2\varepsilon+1}(v)$. 
    A tree $G^v$ for the node $v$ is an \emph{imagined tree} consistent with $S_G^v$ if, under the same graph sketching procedure, the sketch that $v$ obtains in $G^v$ equals $S_G^v$.
    
\end{definition}

Notice that $2$ neighboring nodes $u,v$ may imagine different imagined trees $G_u,G_v$. For example, in Figure~\ref{fig:brothers} node $v$ may imagine the tree on the right, while node $u$ may imagine the tree on the left (which in this case is the original tree).

The (actual) neighborhood $N_{2\varepsilon+1}(v)$ in $G$ is one imagined tree, and in general, there may be several, since $S_G^v$ need not determine the neighborhood uniquely. 
However, $G^v$ and $N_{2\varepsilon+1}(v)$ in $G$ have the same number of nodes, as the sketch encodes the number of nodes in $N_{2\varepsilon+1}(v)$. 
Each imagined tree carries certificates that the node actually reads and uses to make further decisions. 

\begin{definition}[\textbf{Sketch Certificates}]
    Consider an imagined tree $G^v$ consistent with $S_G^v$. The certificate assignment that $G^v$ carries is called the sketch certificate assignment and denoted $C_{sk}^v$. The value $C_{sk}^v(u)$ is the sketch certificate of node $u$ in $G^v$. 
\end{definition}

We now define the notion of \emph{imagined certificates} from the perspective of the node $v$, which captures the notion that each $v$ looks for corrections to the sketch certificates on $G^v$.

\begin{definition}[\textbf{Imagined Certificates}]
    \label{def:imagine}
        For node $v$, a function $C_{im}^v$ on the nodes of $G^v$ is called an imagined certificate assignment function, if $C_{im}^v$ is the certificate assignment function that causes $v$ to accept, i.e., $C_{im}^v$ differs from $C_{adv}$ in at most $\varepsilon$ nodes within $N_{2\varepsilon+1}(v)$ and the base algorithm $\mathcal{A}_0$ makes all nodes within $N_{2\varepsilon}(v)$ accept using $C_{im}^v$.
        To refer to this, we sometimes say that $v$ imagines $C_{im}^v$ (this is local from a node's perspective and may not extend to a global one).
\end{definition}

In words, $v$ considers every imagined tree $G^v$ consistent with its sketch, and in each one searches for an imagined certification, i.e., an override of the certificates of at most $\varepsilon$ nodes within its $2\varepsilon+1$-hop neighborhood under which $\mathcal{A}_0$ makes every node accept throughout $N_{2\varepsilon}(v)$.
If some imagined tree admits such a certification, $v$ accepts; otherwise, it rejects. 
If the adversary introduced at most $\varepsilon$ errors, then the actual $N_{2\varepsilon+1}(v)$ in $G$ is one of these imagined trees, and correcting those errors is a valid certification on it, so $v$ accepts. 

The check is limited to $N_{2\varepsilon}(v)$ (rather than $N_{2\varepsilon+1}(v)$) because modifying the certificate of some node $u$ in $v$'s local computation may affect the decisions of $u$ and its neighbors under $\mathcal{A}_0$, and nodes on the boundary $N_{2\varepsilon+1}(v) \setminus N_{2\varepsilon}(v)$ may have neighbors whose labels are unknown to $v$ (and hence it cannot be sure about their decision).

\begin{algorithm2e}
    \SetKwBlock{Halt}{}{}
\newcommand{\haltaccept}{\textbf{halt and accept}}
\newcommand{\haltreject}{\textbf{halt and reject}}    

    \ForEach{imagined tree $G^v$ consistent with $S_G^v$}
    {
    \ForEach{certificate assignment $C_{im}^v$ on $G^v$ differing from $C_{sk}^v$ in at most $\varepsilon$ nodes within $N_{2\varepsilon+1}(v)$}{
        \If{all nodes within $N_{2\varepsilon}(v)$ accept according to $\mathcal{A}_0$ under $C_{im}^v$}{
            \textbf{halt and accept}
        }
    }

    }

    \textbf{halt and reject}

	\caption{\framework{}($\mathcal{A}_0$,$S_G^v$, $\varepsilon$) for node $v$}
 \label{alg:framework}
\end{algorithm2e}

\subsection{Correctness of Algorithm \framework{}}
\label{subsec:analysis}

Here we analyze the correctness of \framework{}. We state the following theorem.

\begin{theorem}
\label{thm:correctness-framework}
\framework{} together with $C_{adv}$ verifies \ul{} in $O(\varepsilon^2)$ CONGEST rounds, even when $C_{adv}$ differs from $C_O$ in at most $\varepsilon$ nodes within $N_{2\varepsilon+1}(v)$ for every node $v$.


\end{theorem}

Completeness is straightforward, and we prove it in Lemma~\ref{lemma:graph-with-property-prediction-model} below. 
The major one is soundness, and its difficulty is that adjacent nodes $u$ and $v$ might reconstruct different imagined trees $G^u$ and $G^v$ and choose different imagined certifications $C_{im}^u$ and $C_{im}^v$ on them. 
We begin with completeness, which is immediate, and then turn to soundness. 

\begin{lemma}[Completeness]
    \label{lemma:graph-with-property-prediction-model}
    If $G$ has a unique leader, \framework{} makes every $v$ accept, for every adversarial certificate assignment $C_{adv}$ that differs from $C_O$ in at most $\varepsilon$ nodes within $N_{2\varepsilon+1}(v)$.
\end{lemma}

\begin{proof}
    This is quite straightforward.
    For an arbitrary node $v$, consider the imagined tree to be the exact neighborhood $N_{2\varepsilon+1}(v)$ of $G$ itself, which is consistent with $v$'s sketch. 
    On it, we take the certificate assignment function $C_O$ restricted to $N_{2\varepsilon+1}(v)$.
    Under $C_O$, every node accepts according to the algorithm $\mathcal{A}_0$, so in particular, all nodes within $N_{2\varepsilon}(v)$ accept. This proves the lemma. 
\end{proof}

This settles completeness. For soundness, suppose $G$ does not have a unique leader, yet, for contradiction, \framework{} makes every node accept. 
Then each $v$ imagines a certification $C_{im}^v$ on some imagined tree $G^v$. 
The following lemma is key: adjacent nodes agree on each other's imagined certificates. Its proof is the most involved part of the analysis, so we state it here and use it to prove soundness.
Due to space management, we defer its proof to Section \ref{sec:erroneous-labeling}.

\begin{lemma}
\label{prop:diff_view_dist}
    Let $u$ and $v$ be adjacent nodes in $G$ that imagines trees $G^u$ and $G^v$ and certifications $C_{im}^u$ and $C_{im}^v$ respectively (recall Definition \ref{def:imagined-graph}-\ref{def:imagine}).
    Then $C_{im}^u(u)=C_{im}^v(u)$ and $C_{im}^u(v)=C_{im}^v(v)$.
\end{lemma}

Assuming Lemma~\ref{prop:diff_view_dist}, soundness follows by stitching the local certifications into a single global one.
We emphasize that soundness does not rely on the error bound: even if $C_{adv} = C_O$, some node must reject whenever $G$ lacks a unique leader.

\begin{lemma}[Soundness]
    \label{lemma:graph-without-property-prediction-model}
    When $G$ does not have a unique leader, for any proof $C$,    \framework{} makes at least one node reject.
\end{lemma}

\begin{proof}
    Suppose, for contradiction, that \framework{} makes every node accept. 
    Then each node $v$ imagines a certification $C_{im}^v$ on some imagined tree $G^v$.
    Define a certificate assignment $C$ on $G$ by $C(v) = C_{im}^v(v)$ for each node $v$, i.e., every node receives the certificate it imagines for itself. We show that $\mathcal{A}_0$ makes every node of $G$ accept under $C$. 
    Fix a node $v$. 
    Since \framework{} makes it accept using $C_{im}^v$, the base algorithm $\mathcal{A}_0$ must make every node within $N_{2\varepsilon}(v)$ of $G^v$ accept under $C_{im}^v$.
    By Lemma~\ref{prop:diff_view_dist}, for every neighbor $u$ of $v$ we have $C_{im}^v(u) = C_{im}^u(u) = C(u)$. Thus $C$ agrees with $C_{im}^v$ on all of $N_1(v)$. 
    As $\mathcal{A}_0$ inspects only $N_1(v)$ to decide for $v$, and it makes $v$ accept under $C_{im}^v$, it also makes $v$ accept under $C$.
    Hence, $\mathcal{A}_0$ accepts every node of $G$ under $C$, but $\mathcal{A}_0$ is a correct error-free verifier for \ul{}, so it must reject some node when $G$ does not have exactly one leader -- a contradiction. 
    \end{proof}


Lemma~\ref{lemma:graph-with-property-prediction-model} (completeness) and~\ref{lemma:graph-without-property-prediction-model} (soundness) together establish correctness of Theorem~\ref{thm:correctness-framework}. The CONGEST round complexity is shown in Section~\ref{sec:sketches}.


\section{Graph Sketches in $O(\varepsilon^2)$ CONGEST Rounds}
\label{sec:sketches}


We now present the algorithm that constructs the graph sketches described in Section~\ref{subsec:general-idea-algorithm}. Recall that it may return a sketch rejection instead of a proper sketch if it is clear that the instance (graph and adversarial certification) should be rejected. Otherwise, every node obtains a proper sketch.
For convenience, we use the phrase \emph{acceptable certification} to denote a certificate assignment function under which $\mathcal{A}_0$ makes every node accept.

\paragraph{Overview:} Notice that an acceptable certification on a $(2\varepsilon+1)$-hop neighborhood can be easily encoded if one knows the graph structure -- it is sufficient to point out which node $r$ has the smallest certificate, and the certificates of every other node is determined by the distance to $r$, see Figure~\ref{fig:root}.

\begin{figure}
    \centering
    \includegraphics[scale=0.5]{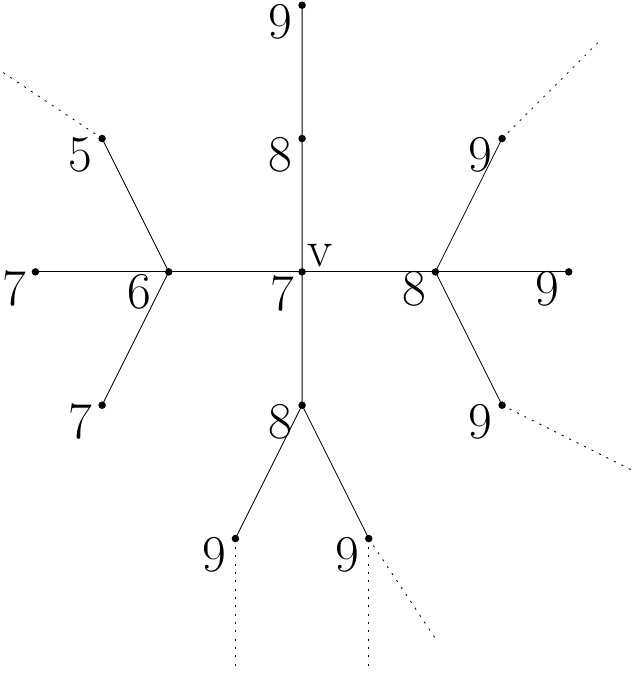}
    \caption{An example graph with constructed acceptable certificates $C'$, from the perspective on node $v$ (at the center). The dotted lines represent edges that $v$ does not know about. In this case, the root $r_v$ will be the node with certificate $C'(r_v)=5$ and every node $w$ must have certificate $C'(w)=C'(r_v)+dist(r_v,w)$ to match the root.}
    \label{fig:root}
\end{figure}

However, we transmit adversarial certificates instead of an acceptable one. If there exists an acceptable certification $C$ of nodes that differs from the adversarial $C_{adv}$ in at most $\varepsilon$ nodes in $(2\varepsilon+1)$-hop neighborhood, then we can transmit $C$ in the compact form above together with the list of deviations from $C_{adv}$, which allow the adversarial certificates to be reconstructed. Otherwise, we can return the sketch rejection, since in that case, the adversary must have violated the constraint of $\varepsilon$ modifications in $(2\varepsilon+1)$-hop neighborhood for every node.
The main challenge is that learning the exact $(2\varepsilon+1)$-hop neighborhood is too expensive in the CONGEST model.
Fortunately, a node do not need the exact structure of the graph to be encoded in the sketch in Section~\ref{subsec:framework}. 

Instead, each node $v$ only needs to know, for each certificate sequence $S$, whether a path through $v$ whose adversarial certificates match $S$ exists.\footnote{For the clarity of presentation, in this paragraph we omit that the position of node $v$, its two neighbors, and possibly the position of the leader are also recorded with each sequence of labels.} Consequently, all the paths with the same certificates only need to be represented by a single sequence.
%
%
In fact, rather than transmitting the list of sequences of certificates, we build a tree that corresponds to a subgraph of $G$ and contains all the unique sequences of certificates and no other sequence.






The algorithm works in phases $j=0,\dots,2\varepsilon$. At the start of a $j$-th phase, each node $v$ knows the tree containing information about its $j$-hop neighborhood. Each node $v$ transmits this tree to each neighbor. Similarly, $v$ receives the trees from each of its neighbor $u$ -- each tree contain information about $j$ hop neighborhood of $u$. Node $v$ then processes these trees to construct a single tree that represents its $j+1$-hop neighborhood that can be used in the next phase.
There are $2\varepsilon+1$ phases, each lasting only $O(\varepsilon)$ CONGEST rounds, for a total of $O(\varepsilon^2)$ CONGEST rounds. 
Next we will present more precise description of the algorithm.

\subsection{Learning Local Neighborhood Subroutine} 
\label{subsec:sketch-construction}
We present the description of a $j$-th phase from the perspective of $v$, ref. to Algorithm \ref{alg:sketch}.

\paragraph{Preliminaries.} As mentioned in the overview, node $v$ knows a subtree $T^v$ of $G$ and its adversarial certification such that for every sequence of labels $S$ that occurs on some path $P$ in $G$ going through $v$ and up to $2$ neighbors of $v$, possibly with a position $i$ within the sequence marked for a node with input $\mathcal{L}$, there exists a path $P'$ in $T^v$ with the same sequence of labels going through $v$ and through the same neighbor(s) of $v$ with a node marked $\mathcal{L}$ (if applies) at the specified position $i$ and vice versa -- for every sequence in $T^v$ there exists a corresponding sequence in $G$.
Furthermore, for each neighbor $u$ of $v$, the node $v$ knows the number $n_u^v$ of nodes at distance $j$ from $v$ whose path to $v$ passes through $u$.

In the base case of $j=0$, each node $v$ knows the tree $T^v$ representing its $0$-hop neighborhood, namely $v$ knows itself and its adversarial label $C_{adv}(v)$. Additionally, $v$ knows $n_u^v = 0$ when $j=0$.
Based on $T^v$, for each neighbor $u$ of $v$, node $v$ prepares new trees $T_{v \rightarrow u}$ that will be transmitted to $u$.\footnote{The trees sent to different neighbors will be almost identical. Thus, this part can be adapted to using broadcast messages, with very little overhead caused by the small differences.} The tree $T_{v \rightarrow u}$ is guaranteed to be small enough that $O(\varepsilon)$ CONGEST rounds will suffice to transmit its structure and encode all the certificates on it.
%
%
An edge $e$ in tree $T_{v \rightarrow u}$ may represent an entire path $P$ in the original tree $T^v$. The weight $w(e)$ of edge $e$ is the number of edges in $P$. Nodes of $T_{v \rightarrow u}$ that are not interior to such a compressed path are called \emph{explicit}.

\paragraph{Construction.} Based on the adversarial certificate assignment of tree $T^v$, node $v$ searches for an imagined certification $C'$ of $T^v$. If no such imagined certification exists, then $v$ rejects; this is a sketch rejection. Otherwise, $v$ finds the node $r$ with the smallest certificate in $T^v$.

The tree $T_{v \rightarrow u}$ contains the node $v$, the neighbor $u$ of $v$ with the edge $\{u,v\}$ (with weight $1$), node $r$ and all the nodes and edges (with weight $1$ each) between $v$ and $r$. Furthermore, we add and connect to tree $T_{v \rightarrow u}$ each node $w \in T^v$ such that $C'(w) \neq C_{adv}(w)$.\footnote{Connecting this node may add an edge with weight greater than $1$.}

Connecting a node $u$ to the tree $T_{v \rightarrow u}$ is done by finding the shortest path $P$ in $T^v$ between $u$ and the nearest node $w$ represented (explicit or not) in $T_{v \rightarrow u}$. If $w$ is explicit in $T_{v \rightarrow u}$, then we simply add an edge $e=\{u,w\}$ with weight $w(e)$ equal to the length of $P$. If $w$ is not explicit, then $w$ lies on a path connecting two explicit nodes in $T_{v \rightarrow u}$, which is represented by a single edge $e$ with $w(e)>1$. In this case we add the node $w$ to the tree $T_{v \rightarrow u}$, splitting the edge $e$, and we connect node $u$ to node $w$.

Finally, for each explicit node $w$, we also store the length $L_w$ of the longest path from node $w$ within $T^v$ that does not have any edges represented in tree $T_{v \rightarrow u}$. This length $L_w$ represents a path of nodes $x$ in $T^v$ with adversarial certificates $C_{adv}(x)=C_{adv}(w)+dist(w,x)$.
This ends the construction of the tree $T_{v \rightarrow u}$.
In addition to the tree $T_{v \rightarrow u}$, node $v$ also sends to each neighbor $u$ the number $n_{v \rightarrow u}$ of nodes at distance $j$ from $v$ that lie on the side of edge $\{u,v\}$ that $v$ is on. The value $n_{v \rightarrow u}$ is computed as the sum of $n_w^v$ over all neighbors $w \neq u$ of $v$, plus $1$ to account for node $v$ itself.

\paragraph{Learning the start-of-phase tree.}

We now need to show that what a node receives by the end of phase $j$ is sufficient to build the tree $T$ it needs in the beginning of phase $j+1$. We will show that the tree $T$ we build (i) does not encode any sequence that does not exist in $G$ and (ii) encodes all sequences that exist in $G$.

Let us first build $T$. During phase $j$, node $v$ receives trees $T_{u \rightarrow v}$ from each neighbor $u$. Each tree $T_{u \rightarrow v}$ encodes a subtree of the tree $T^u$ that node $u$ started phase $j$ with. Trees $T^u$ are themselves subtrees of $G$. Thus, $v$ can merge trees $T_{u \rightarrow v}$ into a subtree $T$ of $G$.\footnote{Since positions of $v$ and $u$ are marked in each $T_{u \rightarrow v}$, then parts of the trees that overlap are correctly merged.}

Showing (i), i.e., all the sequences encoded in $T$ are also encoded in $G$, is trivial. Since $T$ is a subtree of $G$, then for every sequence in $T$ there exists a corresponding sequence in $G$. 

It remains to show (ii), i.e., no sequences in $G$ were omitted during our construction. Take any sequence $S$ that corresponds to some path $P$ in $G$ passing through $v$ and two\footnote{A path passing through only $1$ neighbor of $v$ is analyzed analogously.} neighbors $u_1,u_2$ of $v$ such that $P$ lies within $j+1$-hop neighborhood of $v$. We will show that the sequence $S$ also occurs in $T$. We split $S$ into two subsequences around the position of $v$: $S_1$ containing certificates before position of $v$ and certificate of $v$ itself and $S_2$ containing certificates after position of $v$ and certificate of $v$ itself. Without loss of generality assume that $S_1$ contains the position corresponding to $u_1$ and $S_2$ contains the position corresponding to $u_2$.

Notice that subsequence $S_1$ corresponds to nodes within $j$-hop neighborhood of $u_1$, thus $S_1$ occurs in $T^{u_1}$.
Now we show that $S_1$ is encoded in $T_{u_1 \rightarrow v}$ (proof for $S_2$ is analogous). Consider path $P_1$ in $T^{u_1}$ that corresponds to sequence $S_1$. If $P_1$ fully overlaps with some path on $T_{u_1 \rightarrow v}$, then $S_1$ is trivially encoded by $T_{u_1 \rightarrow v}$. Otherwise, path $P_1$ only partially overlaps with nodes in $T_{u_1 \rightarrow v}$ and at some point branches off the nodes represented in $T_{u_1 \rightarrow v}$. Let $w$ be the furthest node from $v$ on $P_1$ that is explicit in $T_{u_1 \rightarrow v}$. Consider any node $x$ on $P_1$ after $w$, i.e., such that $dist(x,v)>dist(x,w)$. Notice that $C_{adv}(x)=C'(x)$ as otherwise $x$ would be added and connected to the tree $T_{u_1 \rightarrow v}$. Thus, $C_{adv}(x)=C'(x)=C'(r)+dist(x,r)=C'(w)+dist(w,x)$, which is represented in $T_{u_1 \rightarrow v}$ by the length $L_w \geq dist(w,x)$ of the branch at $w$. Therefore, there exists a path in $T_{u_1 \rightarrow v}$ and a length $L_w$ of the branch at $w$ that together encode a path corresponding to the sequence $S_1$.


Thus, when node $v$ merges $T_{u_1 \rightarrow v}$ and $T_{u_2 \rightarrow v}$ (and other trees), it creates a path in $T$ that corresponds to $S$, which completes the proof of (ii).




\subsection{Round Complexity}

In this section, we calculate how many CONGEST rounds are required to perform all the transmissions described above.
Each node $v$ transmits the tree $T_{v \rightarrow u}$ and the number $n_{v \rightarrow u}$ to each neighbor $u$. The tree $T_{v \rightarrow u}$ is composed of
\begin{itemize}[leftmargin=*]
\setlength{\itemsep}{-1pt}
    \item path from $v$ to $r$, of length at most $j$ and the label $C'(r)$
    \item for every node $w$ such that $C_{adv}(w)\neq C'(w)$, we put at most $2$ nodes and edges in $T_{v \rightarrow u}$ and we store $C_{adv}(w)$
    \item for every explicit node $w$ in $T_{v \rightarrow u}$ we store the length of the longest branching from $w$ that does not go through any edge in $T_{v \rightarrow u}$
\end{itemize}

In total, node $v$ transmits up to $N=(j+1)+(2\varepsilon)$ nodes, $N-1$ edges with their weights, $1+\varepsilon$ certificates and additional $N$ numbers (lengths of branches) to transmit the tree $T_{v \rightarrow u}$. Furthermore, node $v$ transmits to each neighbor $u$ the number $n_{v \rightarrow u}$. In total, $O(\varepsilon)$ CONGEST messages are sufficient to make all the transmissions in $j$-th phase, which proves the round complexity in Theorem~\ref{thm:correctness-framework}.


\subsection{Decoding the Graph Sketch}

In this section we describe how to obtain the information necessary in Section~\ref{subsec:framework} from our construction. Recall the information we needed:
\begin{enumerate}[leftmargin=*]
\setlength{\itemsep}{-1pt}
    \item The number of nodes in $N_{2\varepsilon+1}(v)$ in $G$. \label{item:sketch-number-of-nodes}
    
    \item For each $u \in N_1(v)$, the adversarial label $C_{adv}(u)$ together with the port at $v$ leading to node $u$.%
    \label{item:sketch-neighbors}
    

    \item For every sequence of certificates $(c_1, c_2, \dots, c_k)$ with $k \le 4\varepsilon+3$, every pair of positions $m,j \in \{1, \dots, k\}$, and every pair $u_1, u_2$ of neighbors of $v$, the node $v$ knows whether the sequence is \emph{realizable at positions $m,j$ through $u_1, u_2$}: that is, whether there exists a path $P = (v_1, v_2, \dots, v_k)$ contained entirely within $N_{2\varepsilon+1}(v)$ such that \label{item:sketch-sequence} (i) $v$ occupies position $m$ in $P$, i.e.\ $v = v_m$, (ii)  a node $x$ with input $\mathcal{L}$ occupies position $j$ in $P$, i.e.\ $x = v_j$; if such a node does not exist then $j = \bot$, (iii) $u_1 = v_{m-1}$ whenever $m \ge 2$, and $u_2 = v_{m+1}$ whenever $m \le k-1$ (the path enters and leaves $v$ through the specified neighbors), and (iv) the adversarial certificates along $P$ match the sequence, i.e.\ $C_{adv}(v_i) = c_i$ for every $i$.

\end{enumerate}

After $2\varepsilon+1$ phases described in Subsection~\ref{subsec:sketch-construction}, each node $v$ knows a subtree $T$ of $G$ and its adversarial certificate assignment such that for every unique sequence of labels $S$ that occurs on some path $P$ in $G$ going through $v$ and up to $2$ neighbors of $v$, possibly with a position $i$ within the sequence marked for a node with input $\mathcal{L}$, there exists a path $P'$ in $T$ with the same sequence of labels going through $v$ and through the same neighbor(s) of $v$ with a node marked $\mathcal{L}$ (if applies) at the specified position $i$.

Furthermore, each node $v$ knows for each of its neighbors $u$, the number $n_{u \rightarrow v}$ of nodes at distance $j$ from $v$ that have node $u$ on the path towards $v$, for every $1 \leq j \leq 2\varepsilon+1$.

Thus, each node $v$ can calculate the graph sketch necessary in Section~\ref{subsec:framework}:
\begin{enumerate}
    \item The number of nodes in Item~\ref{item:sketch-number-of-nodes} is the sum of numbers $n_{u \rightarrow v}$ for all neighbors $u$ of $v$ across all the phases.
    \item The adversarial certificate and port number of each neighbor $u$ of $v$ in Item~\ref{item:sketch-neighbors} was learned in phase~$0$.
    \item The tree $T^v$ at each node $v$ known after the last phase\footnote{It may be convenient to think it is the tree known at the start of phase $2\varepsilon+1$, i.e., an artificial phase after the last phase.} contains all the same sequences of certificates (with positions of node $v$, its neighbors and possibly a node with input $\mathcal{L}$) as in the original graph $G$. Thus, each node $v$ can check for each sequence of certificates (with the positions of respective nodes of interest) in $T^v$ to know Item~\ref{item:sketch-sequence}.
\end{enumerate}

\begin{algorithm2e}
$T^v \gets$ the single node $v$ with certificate $C_{adv}(v)$\\
$n_u^v \gets 0$ for every neighbor $u$ of $v$\\

\For{$j = 0$ \KwTo $2\varepsilon$}{
  find a certificate assignment $C'$ on $T^v$ differing from the adversarial certificates of $T^v$ on at most $\varepsilon$ nodes, under which $\mathcal{A}_0$ makes every node of $T^v$ accept\\
  \lIf{no such $C'$ exists}{\textbf{halt and reject} (sketch rejection)}
  $r \gets$ the node of $T^v$ with the smallest certificate under $C'$\\
  \ForEach{neighbor $u$ of $v$}{
    $T_{v \rightarrow u} \gets$ the path from $v$ to $r$ in $T^v$, plus $u$ with the edge
    $\{u,v\}$, each edge of weight~$1$\\
    \ForEach{node $x$ of $T^v$ with $C'(x) \neq C_{adv}(x)$}{
      $y \gets$ the nearest node to $x$ represented in $T_{v \rightarrow u}$; if needed, make $y$
      explicit, splitting its weighted edge\\
      add $x$ and the edge $\{x,y\}$ with weight $dist_{T^v}(x,y)$; store
      $C_{adv}(x)$\\
    }
    \ForEach{explicit node $w$ of $T_{v \rightarrow u}$}{
      store the length of the longest path of $T^v$ from $w$ avoiding the edges of $T_{v \rightarrow u}$
    }
    send $\bigl(T_{v \rightarrow u},\ n_{v \rightarrow u} = 1+\sum_{w \neq u} n_w^v\bigr)$ to $u$\\
  }
  receive a tree and a count from every neighbor $u$; merge each received
  tree into $T^v$ by identifying $v$ and the edge $\{u,v\}$; update $n_u^v$ to
  the received count\\
}
return $T^v$ and $n_u^v$ for each neighbor $u$ of $v$ as the encoding of the sketch $S_G^v$
\caption{\textsc{sketch-construction}($\varepsilon$) at node $v$}
\label{alg:sketch}
\end{algorithm2e}

\section{Proof of Lemma \ref{prop:diff_view_dist}}
\label{sec:erroneous-labeling}

This section is devoted to the proof of Lemma \ref{prop:diff_view_dist}, the one component of the soundness argument (Lemma \ref{lemma:graph-without-property-prediction-model}) we deferred. 
Recall the claim: any two adjacent nodes $u$ and $v$, correcting the certificates they read independently and on separate imagined graphs $G^u$ and $G^v$, must nonetheless agree on each other's imagined certificates -- $C_{im}^u(u) = C_{im}^v(u)$ and $C_{im}^u(v) = C_{im}^v(v)$.
This local agreement is exactly what allows the individual imagined certifications to be stitched into a single consistent assignment on $G$ in the soundness proof.

The major issue is that adjacent nodes $u$ and $v$ might reconstruct different imagined trees $G^u$ and $G^v$ and choose different imagined certifications $C_{im}^u$ and $C_{im}^v$ on them. 
To relate what $u$ and $v$ imagine, we require a common ground, and the following notion of \emph{brothers} lets you do that by matching paths in $G^u$ and $G^v$ to a path in $G$.


Fix a node $w$ and let $P = (v_1, v_2, \dots, v_k)$ be a path in $G$ that passes through $w$ lying within $N_{2\varepsilon+1}(w)$.
By item~\ref{itm:one} of the sketch, $G_w$ contains at least one path $X = (x_1, x_2, \dots, x_k)$ carrying the same certificate sequence, i.e., $C_{sk}^w(x_j) = C_{adv}(v_j)$ for every $j$, with $w$ in the recorded position. We fix one such $X$ and call it a \emph{brother} of $P$ in $G_w$, and we call $x_j$ the brother of $v_j$ for every $j$.
The imagined certification $C_{im}^w$ assigns a certificate to each $x_j$ and let $a_j^w = C_{im}^w(x_j)$ be the certificates that $w$ imagines along $P$. The certificates $a_j^w$ satisfy two properties.
\begin{itemize}
    \item[\textbf{(P1)}] $a_j^w = C_{adv}(v_j)$ for all but at most $\varepsilon$ indices $j$.
    Indeed, $X$ lies within $G_w$, so $C_{im}^w$ differs from $C_{sk}^w$ on at most $\varepsilon$ of its nodes and $C_{sk}^w(x_j) = C_{adv}(v_j)$ by construction of $X$.
    \item[\textbf{(P2)}] 
    Under $C_{im}^w$, every node within $2\varepsilon$ hops from $w$ in $G_w$ accepts according to $\mathcal{A}_0$.
    This covers all of $x_1, \dots, x_k$ except possibly the two endpoints $x_1$ and $x_k$, which may fall outside $N_{2\varepsilon}(w)$.
\end{itemize}
We make the notion of brother symmetric and transitive, i.e., 
(i) if $P_u$ in $G^u$ is a brother of $P$ in $G$, then $P$ is also a brother of $P_u$, and
(ii) if $P_1$ is a brother of $P_2$ and $P_2$ is a brother of $P_3$, then $P_1$ is a brother of $P_3$.
This lets us relate paths across two different imagined trees $G^u$ and $G^v$, i.e., if $X_u$ in $G^u$ and $X_v$ in $G^v$ are both brothers of the same path $P$ of $G$, they are brothers of each other. 
We extend this relation to nodes, i.e., if $P$ and $P'$ are brothers, then the $j$-th node of $P$ and the $j$-th node of $P'$ are brothers too, for every $j$. 

\begin{remark}
\label{rem:brothers}
Every path $P_u$ through $u$ in $G^u$ has a brother in (original tree) $G$ and a brother in $G^v$ for every $v \neq u$. Similarly, any node $w$ in $G^u$ must have at least one brother in the original tree $G$ and at least one brother in $G^v$ for every $v \neq u$.
\end{remark}

See Figure~\ref{fig:brothers} for an illustration. Notice that there is no one-to-one relation between a path in $G$ and a path in $G_v$, nor between a node in $G$ and a node in $G_v$.

\begin{figure}
    \centering
    \includegraphics[width=0.65\linewidth]{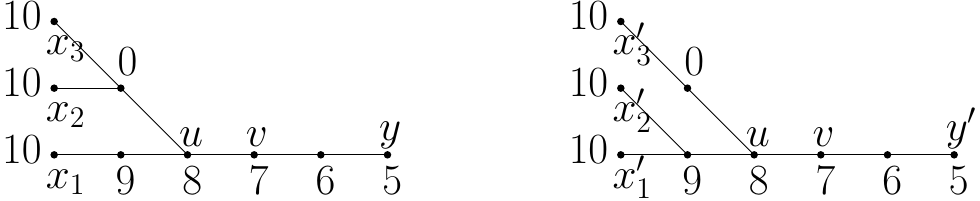}
    \caption{An example illustrating brothers. Let the left graph be the original graph $G$ with the adversarial certificates, while the right graph be the imagined graph $G_v$ with the sketch certificates. Notice that $G$ and $G_v$ have the same sequences of certificates and the same numbers of nodes at every distance $j$ from $v$, i.e., $G_v$ is an imagined graph of $G$. Paths $(x_1,\dots,y)$ in $G$, $(x_1',\dots,y')$ in $G_v$ and $(x_2',\dots,y')$ in $G_v$ are brothers of each other. Similarly, Paths $(x_2,\dots,y)$ in $G$, $(x_3,\dots,y)$ in $G$ and $(x_3',\dots,y')$ in $G_v$ are brothers of each other. Furthermore, nodes $x_1$ in $G$, $x_1'$ in $G_v$ and $x_2'$ in $G_v$ are brothers. 
    }
    \label{fig:brothers}
\end{figure}

With brothers in hand, we can compare what different nodes imagine.

The proof of Lemma~\ref{prop:diff_view_dist} combines two things. The first is a property of certifications accepted by $\mathcal{A}_0$, which is a distance labeling from its unique smallest-certificate node, its \emph{root}, so that once the root is located, every certificate is determined by a distance to it.
The second is the uniqueness of accepting certification along a path with fixed endpoints (Claim \ref{lem:path_labeling}), which allows us to transfer agreement from a few matched positions to the entire path. We first establish this claim, then prove the lemma by case analysis.


    \begin{remark}[Acceptable certification]
        \label{remark:labels-differ-by-1}
        Under any certificate assignment function $C'$ acceptable by $\mathcal{A}_0$, each node $v$ with certificate $C'(v)$ must have at most\footnote{If $C'(v)=0$, then there is no parent.} one adjacent node with certificate $C'(v)-1$ and all other neighbors must have certificate $C'(v)+1$.
    \end{remark}

    Remark~\ref{remark:labels-differ-by-1} follows directly from the description of the base algorithm $\mathcal{A}_0$.

    \begin{claim}
    \label{lem:path_labeling}
    
        Consider a path $P = (v_0, v_1, \ldots, v_k)$ and any fixed $a,b \in \mathbb{N}\cup \{0\}$.
        Then, there is at most one certificate assignment function $C$ on $P$ such that $C(v_0) = a$, $C(v_k) = b$, under which every node of $P$ accepts according to $\mathcal{A}_0$.
    \end{claim}

\begin{proof}[Proof of Claim \ref{lem:path_labeling}]

If no accepting certification with the prescribed endpoints exists, then the claim holds trivially, so assume at least one does, and we show that it is unique. 
We argue by induction on the path length $k$.

        \textit{Inductive hypothesis $H(k')$ for $1 \leq k'<k$:} For any $a',b' \in \mathbb{N}\cup \{0\}$, for any path $P'$ of length $k'$ with endpoint nodes assigned the certificates $a'$ and $b'$, there exists at most one certificate assignment function $C$ under which every node on $P'$ accepts according to the algorithm $\mathcal{A}_0$.

        \textit{Base case $H(1)$:} A path $(u_0, u_1)$ with fixed endpoint labels $C(u_0) = a'$ and $C(u_1) = b'$ either satisfies the acceptance conditions of $\mathcal{A}_0$ or not. In either case, at most one certificate assignment function exists.
        Thus $H(1)$ holds.

        \textit{Inductive step ($H(k') \Rightarrow H(k'+1)$):}
        Let $P'=(u_0,u_1,\dots,u_{k'+1})$ be a path of length $k'+1$ with endpoint certificates $a',b'$ where $a' \neq 0$ or $b' \neq 0$. 
        If no proof satisfies these constraints while making all nodes accept under $\mathcal{A}_0$, then $H(k'+1)$ holds trivially. 
        
        Otherwise, let $C'$ be such a certificate assignment function on $P'$ and let $m = \max\{C'(u_0),\dots,$ $C'(u_{k'+1})\}$.
        We show that no internal node $u_j$ ($1 \leq j \leq k'$) can have certificate $m$.
        Suppose $C'(u_j) = m$ for some $1 \leq j \leq k'$. By Remark~\ref{remark:labels-differ-by-1}, node $u_j$ requires at most $1$ neighbor with the certificate $m-1$ and all others with $m+1$. 
        But both neighbors $u_{j-1}$ and $u_{j+1}$ satisfy $C'(u_{j-1}) \leq m$ and $C'(u_{j+1}) \leq m$ by maximality, so neither has the certificate $m+1$, which contradicts Remark~\ref{remark:labels-differ-by-1}. 
        Therefore, no internal node of $P'$ can have the certificate $m$.

        Consequently, $m$ appears only at endpoints of $P'$. Without loss of generality, let $C'(u_0)=m=a'$.
        Remark~\ref{remark:labels-differ-by-1} forces $C'(u_1) \in \{m-1, m+1\}$, and by maximality $C'(u_1) = m-1 = a'-1$.
        Deleting $u_0$ leaves the path $(u_1, \dots, u_{k'+1})$ of length $k'$ with fixed endpoint certificates $a'-1$ and $b'$.
        By $H(k')$ applied to  $(u_1, \dots, u_{k'+1})$ with endpoint labels $a'-1$ and $b'$, at most one proof of this subpath exists.
        Hence, at most one proof on $P'$ exists, establishing $H(k'+1)$.

        This completes the proof of Claim \ref{lem:path_labeling}.
\end{proof}

    We are now ready to prove Lemma~\ref{prop:diff_view_dist}.




\begin{proof}[Proof of Lemma~\ref{prop:diff_view_dist}]

Let $T_u$ denote the subgraph of $G^u$ induced by the nodes that are $2\varepsilon+1$ hops away from both $u$ and $v$, i.e., $N_{2\varepsilon+1}(u) \cap N_{2\varepsilon+1}(v)$ in $G^u$. 
If a node $w \in T_u$ is exactly $2\varepsilon+1$ hops away from either $u$ or $v$, then $w$ is called an \emph{outer node}, otherwise $w$ is an \emph{inner node}. Let node $r_u \in T_u$, called \emph{root}, be the node with the smallest imagined certificate, i.e., $C_{im}^u(r_u) = \min_{w \in T_u}C_{im}^u(w)$.\footnote{The root of $T_u$ does not have to be the leader. There may be no leader inside $T_u$, but $T_u$ still has the node with the smallest certificate and the certificate may be greater than $0$.}\footnote{There is exactly $1$ such node $r_u$. If there were multiple such nodes $r_u$ and $r_u'$, then somewhere on the path between them there would be a node $w$ that either has $2$ parents or that is adjacent to a node with the same label; in either case the node $w$ rejects according to $\mathcal{A}_0$, which contradicts our assumption that $C_{im}(u)$ is an imagined certification (imagined certification must be accepted by $\mathcal{A}_0$).}
We define analogously $T_v$ as a subgraph of $G^v$, inner and outer nodes in $T_v$, and a root node $r_v$ in $T_v$.

\begin{remark}
\label{rem:increasing_path}
    If $r_u$ is a root of $T_u$, then for any node $w \in T_u$ we have $C_{im}^u(w)=C_{im}^u(r_u)+dist(r_u,w)$.
\end{remark}

The remark above follows immediately from the fact that an imagined certification is accepted by $\mathcal{A}_0$, which means that the certificates increase by $1$ with each hop away from the leader.

We consider the following cases:
\begin{itemize}
    \item \textbf{Case 1:} $r_u$ is an inner node
    \item \textbf{Case 2:} $r_u$ is an outer node
\end{itemize}

\textbf{Case 1:} Since $r_u$ has the smallest certificate in $T_u$, then $r_u$ does not have a parent in $T_u$. The only nodes without parents in an acceptable certification must have certificate $0$, so $C_{im}^u(r_u)=0$. It follows that $C_{im}^u(u)=dist(r_u,u)$ and $C_{im}^u(v)=dist(r_u,v)$.

If a node in an acceptable certification has label $0$, then it must be the leader -- in this case, it means that node $r_u$ corresponds to the $\mathcal{L}$ position in some sequence of certificates in the graph sketch. Thus, according to Remark~\ref{rem:brothers}, $r_u$ is a brother of a node $r$ in $G$ such that $r$ received input $\mathcal{L}$. 

Furthermore, $r_u$ has a brother $w$ in $G^v$ in the same position relative to nodes $u$ and $v$, so $w$ is also an inner node in $T_v$. Since $w$ is a brother of $r_u$, then it corresponds to the $\mathcal{L}$ position, thus $C_{im}^v(w)=0$. Again, we get that $C_{im}^v(u)=dist(w,u)$ and $C_{im}^v(v)=dist(w,v)$. Since $w$ and $r_u$ occupy the same position in the corresponding sequence of labels in the graph sketch, then they are at the same distance from $u$ and they are at the same distance from $v$. Thus, $C_{im}^u(u)=dist(r_u,u)=dist(w,u)=C_{im}^v(u)$ and $C_{im}^u(v)=dist(r_u,v)=dist(w,v)=C_{im}^v(v)$, which proves the lemma in this case.

\textbf{Case 2:}  In this case, let us fix a path $P_u$ in $T_u$ that starts in $r_u$ and ends in $u$ or $v$ (whichever is further). Notice that $P_u$ contains $2\varepsilon+2$ nodes.\footnote{By definition of an inner node, $r_u$ is at distance $2\varepsilon+1$ from $u$ or $v$. Together with the node $u$ or $v$, this constitutes a path of length exactly $2\varepsilon+2$.} According to Remark~\ref{rem:brothers}, the path $P_u$ has a brother $P_v$ in $G^v$. Since they are brothers, they correspond to the same sequence $S$ of sketch labels.\footnote{Note that the brother $w$ in $P_v$ of root $r_u$ in $P_u$ does not necessarily have to be the root in $P_v$. They have the same sketch labels, but their imagined labels may differ significantly.}
Node $u$ modified up to $\varepsilon$ labels in $S$ to obtain $C_{im}^u$ on $P_u$ and node $v$ modified up to $\varepsilon$ labels in $S$ to obtain $C_{im}^v$ on $P_v$. In total up to $2\varepsilon$ modifications were introduced to $S$ out of $2\varepsilon+2$ labels in $S$. Thus, there are at least $2$ nodes $u_1,u_2 \in P_u$ with unchanged labels that have brothers $v_1,v_2 \in P_v$ also with unchanged labels, i.e., $C_{im}^u(u_1)=C_{sk}^u(u_1)=C_{sk}^v(v_1)=C_{im}^v(v_1)$ and $C_{im}^u(u_2)=C_{sk}^u(u_2)=C_{sk}^v(v_2)=C_{im}^v(v_2)$.  

According to Remark~\ref{rem:increasing_path}, since $r_u$ has the smallest label on $P_u$, we get that $C_{im}^u(x)=C_{im}^u(r_u)+dist(r_u,x)$ for all $x \in P_u$. In particular, $C_{im}^u(u_1)=C_{im}^u(r_u)+dist(r_u,u_1)$ and $C_{im}^u(u_2)=C_{im}^u(r_u)+dist(r_u,u_2)$. Their brothers $v_1$ and $v_2$ in $P_v$ have the same labels and positions as $u_1$ and $u_2$ in $P_u$. Without loss of generality assume $u_1$ is closer than $u_2$ to $r_u$. Thus, $C_{im}^v(v_1) = C_{im}^u(u_1) < C_{im}^u(u_2)= C_{im}^v(v_2)$. According to Claim~\ref{lem:path_labeling}, there is at most one acceptable certificate assignment function of nodes on the subpath $P_v' \subseteq P_v$ that starts in $v_1$ with label $C_{im}^v(v_1)$ and ends in $v_2$ with label $C_{im}^v(v_2)$. Notice that the certificate assignment function $C_{im}^v$ on the nodes starting from $v_1$, through $v_2$ and ending at $u$ or $v$ (whichever is further) must be the same as the certificate assignment function $C_{im}^u$ on the brothers of those nodes in $P_u$. In particular, $C_{im}^u(u)=C_{im}^v(u)$ and $C_{im}^u(v)=C_{im}^v(v)$, which proves the lemma in this case.

In both cases, the lemma follows.
\end{proof}

\begin{figure}[t]
    \centering
    \includegraphics[width=0.5\linewidth]{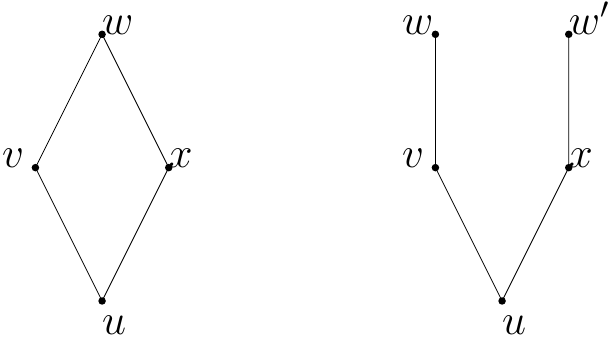}
    \caption{An example of a situation where node $u$ cannot decide whether nodes $w$ and $w'$ it learns about from $v$ and $x$ respectively are actually the same node. Thus, node $u$ cannot decide what is the exact structure of the graph.}
    \label{fig:cycle}
\end{figure}

\section{Future Works and Concluding Remarks}
\label{sec:conclusion}

We initiated the study of error-resilient local verification in the
bandwidth-restricted \textsc{Congest} model, using \ul{} on
trees as a first testbed. Our central message is that error-resilience does not
require a node to reconstruct its full $(2\varepsilon+1)$-hop neighborhood: a
compact $O(\varepsilon^2\log n)$-bit \emph{local graph sketch} captures exactly
the information the verifier consumes, letting each node decide via
\framework{} after only $O(\varepsilon^2)$ rounds of
\textsc{sketch-construction}. 

The most immediate open direction is to generalize beyond trees to arbitrary
graphs. The obstacle is cycles: as illustrated in Figure~\ref{fig:cycle}, when
two neighbors $v$ and $x$ each report information about a common node $w$ to $u$,
node $u$ cannot tell whether the two images refer to the same node, and
transmitting node identifiers to disambiguate is too costly in \textsc{Congest}.
Resolving this identity ambiguity within the sketch framework is the key
challenge for the general case.
In fact, a larger graph can be constructed to generalize this difficulty. 
The compression technique can be modified to distinguish nodes based on their IDs. 

More broadly, we believe the local graph sketch technique---and the accompanying
notions of imagined trees and imagined certifications---are of independent
interest. Summarizing a neighborhood compactly enough to compute on it, without
learning its exact structure, is a natural primitive for communication-efficient,
error-resilient certification, and we expect these ideas to extend to other
fundamental verification problems (cycle-freeness, spanning trees, and beyond) as
well as to the advice/prediction model, where auxiliary information may likewise
be imperfect.


\begin{table}[h!]
\centering
\renewcommand{\arraystretch}{1.15}
\begin{tabular}{|l|l|}
\hline
\textbf{Variable} & \textbf{Description}\\
\hline
$u,v, \dots, y$ & Node variables (possibly with superscripts and subscripts)\\ \hline

$a,b,c$ & Certificate variables (possibly with superscripts and subscripts)\\ \hline

$C$ & Certificate assignment function (possibly with superscripts and subscripts) \\ \hline

$i,j,k$ & Index variables \\ \hline

$T$ & Tree variable (possibly with superscripts and subscripts) \\ \hline

$P$ & Path variable (possibly with superscripts and subscripts) \\ \hline

$S$ & Set or Sequence (of certificates) variable (possibly with superscripts and subscripts) \\ \hline

$\mathcal{A}$ & Verifier (possibly with superscripts and subscripts)\\
\hline
\end{tabular}
\caption{List of Variables}
\label{table:variables}
\end{table}

\subsection*{Acknowledgment:} This work is supported by Polish National Science Centre project no.\ 2020/39/B/ST6/03288.

\bibliographystyle{plain}
\bibliography{bibliography}

@article{DBLP:journals/dc/KormanK07,
  author       = {Amos Korman and
                  Shay Kutten},
  title        = {Distributed verification of minimum spanning trees},
  journal      = {Distributed Comput.},
  volume       = {20},
  number       = {4},
  pages        = {253--266},
  year         = {2007},
  url          = {https://doi.org/10.1007/s00446-007-0025-1},
  doi          = {10.1007/S00446-007-0025-1},
  bibsource    = {dblp computer science bibliography, https://dblp.org}
}

@inproceedings{10.1145/1073814.1073817,
author = {Korman, Amos and Kutten, Shay and Peleg, David},
title = {Proof labeling schemes},
year = {2005},
isbn = {1581139942},
publisher = {Association for Computing Machinery},
address = {New York, NY, USA},
url = {https://doi.org/10.1145/1073814.1073817},
doi = {10.1145/1073814.1073817},
booktitle = {Proceedings of the Twenty-Fourth Annual ACM Symposium on Principles of Distributed Computing (PODC '05)},
pages = {9–18},
numpages = {10},
location = {Las Vegas, NV, USA},
series = {PODC '05}
}

@article{goos2016locally,
  title={Locally checkable proofs in distributed computing},
  author={G{\"o}{\"o}s, Mika and Suomela, Jukka},
  journal={Theory of Computing},
  volume={12},
  pages={1--33},
  year={2016},
  publisher={University of Chicago}
}

@InProceedings{bousquet_et_al:LIPIcs.DISC.2025.18,
  author =	{Bousquet, Nicolas and Feuilloley, Laurent and Zeitoun, S\'{e}bastien},
  title =	{{Complexity Landscape for Local Certification}},
  booktitle =	{39th International Symposium on Distributed Computing (DISC 2025)},
  pages =	{18:1--18:21},
  series =	{Leibniz International Proceedings in Informatics (LIPIcs)},
  ISBN =	{978-3-95977-402-4},
  ISSN =	{1868-8969},
  year =	{2025},
  volume =	{356},
  editor =	{Kowalski, Dariusz R.},
  publisher =	{Schloss Dagstuhl -- Leibniz-Zentrum f{\"u}r Informatik},
  address =	{Dagstuhl, Germany},
  URL =		{https://drops.dagstuhl.de/entities/document/10.4230/LIPIcs.DISC.2025.18},
  URN =		{urn:nbn:de:0030-drops-248350},
  doi =		{10.4230/LIPIcs.DISC.2025.18}
}

@inproceedings{ostrovsky2017space,
  title={Space-time tradeoffs for distributed verification},
  author={Ostrovsky, Rafail and Perry, Mor and Rosenbaum, Will},
  booktitle={International Colloquium on Structural Information and Communication Complexity},
  pages={53--70},
  year={2017},
  organization={Springer}
}

@article{bousquet2024local,
  title={Local certification of graph decompositions and applications to minor-free classes},
  author={Bousquet, Nicolas and Feuilloley, Laurent and Pierron, Th{\'e}o},
  journal={Journal of Parallel and Distributed Computing},
  volume={193},
  pages={104954},
  year={2024},
  publisher={Elsevier}
}

@inproceedings{10.1145/3732772.3733530,
author = {Boyar, Joan and Ellen, Faith and Larsen, Kim S.},
title = {Brief Announcement: Distributed Graph Algorithms with Predictions},
year = {2025},
isbn = {9798400718854},
publisher = {Association for Computing Machinery},
address = {New York, NY, USA},
url = {https://doi.org/10.1145/3732772.3733530},
doi = {10.1145/3732772.3733530},
booktitle = {Proceedings of the ACM Symposium on Principles of Distributed Computing},
pages = {322–325},
numpages = {4},
location = {Hotel Las Brisas Huatulco, Huatulco, Mexico},
series = {PODC '25}
}

@article{10.1145/2499228,
author = {Fraigniaud, Pierre and Korman, Amos and Peleg, David},
title = {Towards a complexity theory for local distributed computing},
year = {2013},
issue_date = {October 2013},
publisher = {Association for Computing Machinery},
address = {New York, NY, USA},
volume = {60},
number = {5},
issn = {0004-5411},
url = {https://doi.org/10.1145/2499228},
doi = {10.1145/2499228},
journal = {J. ACM},
month = oct,
articleno = {35},
numpages = {26}
}

@InProceedings{10.1007/978-3-319-12340-0_2,
author="Arfaoui, Heger
and Fraigniaud, Pierre
and Ilcinkas, David
and Mathieu, Fabien",
editor="Kratsch, Dieter
and Todinca, Ioan",
title="Distributedly Testing Cycle-Freeness",
booktitle="Graph-Theoretic Concepts in Computer Science",
year="2014",
publisher="Springer International Publishing",
address="Cham",
pages="15--28",
isbn="978-3-319-12340-0"
}

@book{peleg2000distributed,
  title={Distributed computing: a locality-sensitive approach},
  author={Peleg, David},
  year={2000},
  publisher={SIAM}
}

@InProceedings{feuilloley_et_al:LIPIcs.DISC.2018.25,
  author =	{Feuilloley, Laurent and Hirvonen, Juho},
  title =	{{Local Verification of Global Proofs}},
  booktitle =	{32nd International Symposium on Distributed Computing (DISC 2018)},
  pages =	{25:1--25:17},
  series =	{Leibniz International Proceedings in Informatics (LIPIcs)},
  ISBN =	{978-3-95977-092-7},
  ISSN =	{1868-8969},
  year =	{2018},
  volume =	{121},
  editor =	{Schmid, Ulrich and Widder, Josef},
  publisher =	{Schloss Dagstuhl -- Leibniz-Zentrum f{\"u}r Informatik},
  address =	{Dagstuhl, Germany},
  URL =		{https://drops.dagstuhl.de/entities/document/10.4230/LIPIcs.DISC.2018.25},
  URN =		{urn:nbn:de:0030-drops-98146},
  doi =		{10.4230/LIPIcs.DISC.2018.25}
}

@InProceedings{fraigniaud_et_al:LIPIcs.ITCS.2021.28,
  author =	{Fraigniaud, Pierre and Le Gall, Fran\c{c}ois and Nishimura, Harumichi and Paz, Ami},
  title =	{{Distributed Quantum Proofs for Replicated Data}},
  booktitle =	{12th Innovations in Theoretical Computer Science Conference (ITCS 2021)},
  pages =	{28:1--28:20},
  series =	{Leibniz International Proceedings in Informatics (LIPIcs)},
  ISBN =	{978-3-95977-177-1},
  ISSN =	{1868-8969},
  year =	{2021},
  volume =	{185},
  editor =	{Lee, James R.},
  publisher =	{Schloss Dagstuhl -- Leibniz-Zentrum f{\"u}r Informatik},
  address =	{Dagstuhl, Germany},
  URL =		{https://drops.dagstuhl.de/entities/document/10.4230/LIPIcs.ITCS.2021.28},
  URN =		{urn:nbn:de:0030-drops-135679},
  doi =		{10.4230/LIPIcs.ITCS.2021.28}
}

@inproceedings{10.1145/2767386.2767421,
author = {Baruch, Mor and Fraigniaud, Pierre and Patt-Shamir, Boaz},
title = {Randomized Proof-Labeling Schemes},
year = {2015},
isbn = {9781450336178},
publisher = {Association for Computing Machinery},
address = {New York, NY, USA},
url = {https://doi.org/10.1145/2767386.2767421},
doi = {10.1145/2767386.2767421},
booktitle = {Proceedings of the 2015 ACM Symposium on Principles of Distributed Computing},
pages = {315–324},
numpages = {10},
location = {Donostia-San Sebasti\'{a}n, Spain},
series = {PODC '15}
}

@inproceedings{DBLP:conf/wdag/EmekGK22,
  author       = {Yuval Emek and
                  Yuval Gil and
                  Shay Kutten},
  editor       = {Christian Scheideler},
  title        = {Locally Restricted Proof Labeling Schemes},
  booktitle    = {36th International Symposium on Distributed Computing, {DISC} 2022,
                  Augusta, Georgia, USA, October 25-27, 2022},
  series       = {LIPIcs},
  pages        = {20:1--20:22},
  publisher    = {Schloss Dagstuhl - Leibniz-Zentrum f{\"{u}}r Informatik},
  year         = {2022},
  url          = {https://doi.org/10.4230/LIPIcs.DISC.2022.20},
  doi          = {10.4230/LIPICS.DISC.2022.20},
  bibsource    = {dblp computer science bibliography, https://dblp.org}
}

@article{DBLP:journals/algorithmica/FraigniaudMRT24,
  author       = {Pierre Fraigniaud and
                  Pedro Montealegre and
                  Ivan Rapaport and
                  Ioan Todinca},
  title        = {A Meta-Theorem for Distributed Certification},
  journal      = {Algorithmica},
  volume       = {86},
  number       = {2},
  pages        = {585--612},
  year         = {2024},
  url          = {https://doi.org/10.1007/s00453-023-01185-1},
  doi          = {10.1007/S00453-023-01185-1},
  bibsource    = {dblp computer science bibliography, https://dblp.org}
}

@inproceedings{10.1145/3519270.3538416,
author = {Feuilloley, Laurent and Bousquet, Nicolas and Pierron, Th\'{e}o},
title = {What Can Be Certified Compactly? Compact local certification of MSO properties in tree-like graphs},
year = {2022},
isbn = {9781450392624},
publisher = {Association for Computing Machinery},
address = {New York, NY, USA},
url = {https://doi.org/10.1145/3519270.3538416},
doi = {10.1145/3519270.3538416},
booktitle = {Proceedings of the 2022 ACM Symposium on Principles of Distributed Computing},
pages = {131–140},
numpages = {10},
location = {Salerno, Italy},
series = {PODC'22}
}

@article{ESPERET202268,
title = {Local certification of graphs on surfaces},
journal = {Theoretical Computer Science},
volume = {909},
pages = {68-75},
year = {2022},
issn = {0304-3975},
doi = {https://doi.org/10.1016/j.tcs.2022.01.023},
url = {https://www.sciencedirect.com/science/article/pii/S0304397522000354},
author = {Louis Esperet and Benjamin Lévêque}}

@inproceedings{10.1145/3382734.3404505,
author = {Feuilloley, Laurent and Fraigniaud, Pierre and Montealegre, Pedro and Rapaport, Ivan and R\'{e}mila, \'{E}ric and Todinca, Ioan},
title = {Compact Distributed Certification of Planar Graphs},
year = {2020},
isbn = {9781450375825},
publisher = {Association for Computing Machinery},
address = {New York, NY, USA},
url = {https://doi.org/10.1145/3382734.3404505},
doi = {10.1145/3382734.3404505},
booktitle = {Proceedings of the 39th Symposium on Principles of Distributed Computing},
pages = {319–328},
numpages = {10},
location = {Virtual Event, Italy},
series = {PODC '20}
}

@article{FEUILLOLEY20239,
title = {Local certification of graphs with bounded genus},
journal = {Discrete Applied Mathematics},
volume = {325},
pages = {9-36},
year = {2023},
issn = {0166-218X},
doi = {https://doi.org/10.1016/j.dam.2022.10.004},
url = {https://www.sciencedirect.com/science/article/pii/S0166218X22003833},
author = {Laurent Feuilloley and Pierre Fraigniaud and Pedro Montealegre and Ivan Rapaport and Éric Rémila and Ioan Todinca}
}

@article{feuilloley2021introduction,
  title={Introduction to local certification},
  author={Feuilloley, Laurent},
  journal={Discrete Mathematics \& Theoretical Computer Science},
  volume={23},
  number={Distributed Computing and Networking},
  year={2021},
  publisher={Episciences. org}
}

@article{feuilloley2025proving,
  title={Proving there is a leader without naming it},
  author={Feuilloley, Laurent and Sedl{\'a}{\v{c}}ek, Josef Erik and Sl{\'a}vik, Martin},
  journal={arXiv preprint arXiv:2511.15491},
  year={2025}
}

@InProceedings{fraigniaud_et_al:LIPIcs.DISC.2023.20,
  author =	{Fraigniaud, Pierre and Mazoit, Fr\'{e}d\'{e}ric and Montealegre, Pedro and Rapaport, Ivan and Todinca, Ioan},
  title =	{{Distributed Certification for Classes of Dense Graphs}},
  booktitle =	{37th International Symposium on Distributed Computing (DISC 2023)},
  pages =	{20:1--20:17},
  series =	{Leibniz International Proceedings in Informatics (LIPIcs)},
  ISBN =	{978-3-95977-301-0},
  ISSN =	{1868-8969},
  year =	{2023},
  volume =	{281},
  editor =	{Oshman, Rotem},
  publisher =	{Schloss Dagstuhl -- Leibniz-Zentrum f{\"u}r Informatik},
  address =	{Dagstuhl, Germany},
  URL =		{https://drops.dagstuhl.de/entities/document/10.4230/LIPIcs.DISC.2023.20},
  URN =		{urn:nbn:de:0030-drops-191467},
  doi =		{10.4230/LIPIcs.DISC.2023.20}
}

@article{CENSORHILLEL2020112,
title = {Approximate proof-labeling schemes},
journal = {Theoretical Computer Science},
volume = {811},
pages = {112-124},
year = {2020},
note = {Special issue on Structural Information and Communication Complexit},
issn = {0304-3975},
doi = {https://doi.org/10.1016/j.tcs.2018.08.020},
url = {https://www.sciencedirect.com/science/article/pii/S030439751830536X},
author = {Keren Censor-Hillel and Ami Paz and Mor Perry}
}

@inproceedings{ahn2012analyzing,
  title={Analyzing graph structure via linear measurements},
  author={Ahn, Kook Jin and Guha, Sudipto and McGregor, Andrew},
  booktitle={Proceedings of the twenty-third annual ACM-SIAM symposium on Discrete Algorithms},
  pages={459--467},
  year={2012},
  organization={SIAM}
}

@inproceedings{ghaffari2018congested,
  title={Congested clique algorithms for the minimum cut problem},
  author={Ghaffari, Mohsen and Nowicki, Krzysztof},
  booktitle={Proceedings of the 2018 ACM Symposium on Principles of Distributed Computing},
  pages={357--366},
  year={2018}
}

@inproceedings{jurdzinski2018mst,
  title={MST in O (1) rounds of congested clique},
  author={Jurdzi{\'n}ski, Tomasz and Nowicki, Krzysztof},
  booktitle={Proceedings of the Twenty-Ninth Annual ACM-SIAM Symposium on Discrete Algorithms},
  pages={2620--2632},
  year={2018},
  organization={SIAM}
}

@article{garncarek2026distributed,
  title={Distributed Local Verification using Proofs with (out) Errors},
  author={Garncarek, Pawe\l{} and Jurdzinski, Tomasz and Kowalski, Dariusz and Pramanick, Subhajit},
  journal={arXiv preprint arXiv:2603.20831},
  year={2026}
}

@inproceedings{10.1145/3732772.3733518,
author = {Ben-David, Naama and Dzulfikar, Muhammad Ayaz and Ellen, Faith and Gilbert, Seth},
title = {Byzantine Agreement with Predictions},
year = {2025},
isbn = {9798400718854},
publisher = {Association for Computing Machinery},
address = {New York, NY, USA},
url = {https://doi.org/10.1145/3732772.3733518},
doi = {10.1145/3732772.3733518},
booktitle = {Proceedings of the ACM Symposium on Principles of Distributed Computing},
pages = {3–14},
numpages = {12},
location = {Hotel Las Brisas Huatulco, Huatulco, Mexico},
series = {PODC '25}
}

@article{lykouris2021competitive,
  title={Competitive caching with machine learned advice},
  author={Lykouris, Thodoris and Vassilvitskii, Sergei},
  journal={Journal of the ACM (JACM)},
  volume={68},
  number={4},
  pages={1--25},
  year={2021},
  publisher={ACM New York, NY}
}

\end{document}